\documentclass{IEEEtran}

\usepackage{graphicx}
\usepackage{epsfig}
\usepackage{mathptmx}
\usepackage{times}
\usepackage{amsmath} 
\usepackage{amsthm}
\usepackage{amssymb}
\usepackage{tcolorbox}
\usepackage{cite}
\usepackage[ruled,vlined]{algorithm2e}
\usepackage[colorlinks = true, linkcolor = blue, citecolor =
blue]{hyperref}
\usepackage{enumitem}
\usepackage{subcaption}
\usepackage[mathscr]{euscript}
\usepackage{xcolor}

\newtheorem{thm}{Theorem}
\newtheorem{prop}[thm]{Proposition}
\newtheorem{lem}[thm]{Lemma}
\newtheorem{cor}[thm]{Corollary}
\newtheorem{rem}[thm]{Remark}
\newtheorem{definition}[thm]{Definition}

\DeclareMathOperator*{\argmin}{argmin}

\newcommand{\ldef}{:=}
\newcommand{\rdef}{=:}
\newcommand{\Mc}[1]{\mathcal{#1}}
\newcommand{\real}{\ensuremath{\mathbb{R}}}
\newcommand{\complex}{\ensuremath{\mathbb{C}}}
\newcommand{\realp}{\ensuremath{\mathbb{R}_{>0}}}
\newcommand{\realz}{\ensuremath{\mathbb{R}_{\ge 0}}}

\newcommand{\nat}{{\mathbb{N}}}
\newcommand{\natz}{{\mathbb{N}}_0}
\newcommand{\e}{\mathrm{e}}
\newcommand{\dd}{\mathrm{d}}
\newcommand{\cl}[1]{\mathrm{cl}(#1)}
\newcommand{\norm}[1]{\left\lVert #1 \right\rVert}

\newcommand{\trsp}{^\mathrm{T}}
\newcommand{\tmin}{\tau_{\min}}
\newcommand{\tmax}{\tau_{\max}}
\newcommand{\invspace}{\Mc{E}}
\newcommand{\spec}{\sigma}
\newcommand{\rspantxt}{\real\mathrm{-span}}
\newcommand{\rspan}{\Mc{H}}

\newcommand{\reig}{\real\mathrm{-eigen \ subspace}}

\newcommand{\thmtitle}[1]{\mbox{}\textit{(#1).}}

\newcommand{\remend}{\relax\ifmmode\else\unskip\hfill\fi\hbox{$\bullet$}}

\graphicspath{{figs/}}

\begin{document}
	
\title{Analysis of Inter-Event Times in Linear Systems under
  Region-Based Self-Triggered Control} \author{Anusree Rajan, \IEEEmembership{Member, IEEE}, and
  Pavankumar Tallapragada, \IEEEmembership{Member, IEEE}%
  \thanks{This work was partially supported by Science and Engineering
    Research Board under grant CRG/2019/005743. A. Rajan was supported
    by a fellowship grant from the Centre for Networked Intelligence
    (a Cisco CSR initiative) of the Indian Institute of Science.\\
    Anusree Rajan is with the Department of Electrical Engineering,
    Indian Institute of Science and Pavankumar Tallapragada is with
    the Department of Electrical Engineering and Robert Bosch Centre
    for Cyber Physical Systems, Indian Institute of Science {\tt\small
      \{anusreerajan,pavant\}@iisc.ac.in }} }

\maketitle

\begin{abstract}
  This paper analyzes the evolution of inter-event times (IETs) in linear
  systems under region-based self-triggered control (RBSTC). In this control
  method, the state space is partitioned into a finite number of conic
  regions and each region is associated with a fixed IET. In this framework, studying the steady state behavior of the
  IETs is equivalent to studying the existence of a conic
  subregion that is positively invariant under the map that gives the
  evolution of the state from one event to the next. We provide
    necessary conditions and sufficient conditions for the existence of a
    positively invariant subregion (PIS). We also provide necessary and
    sufficient conditions for a PIS to be
    asymptotically stable. Indirectly, they provide necessary and
    sufficient conditions for local convergence of IETs
    to a constant or to a given periodic sequence.  We illustrate the
  proposed method of analysis and results through numerical
  simulations.
\end{abstract}

\begin{IEEEkeywords}
  Self-triggered control, Inter-event times, Networked control systems
\end{IEEEkeywords}
\vspace{-2ex}
\section{Introduction}

Self-triggering is an efficient method for control under
resource constraints. In this method, the control update
times are opportunistic and implicitly determined by a triggering
rule. Thus, understanding inter-event times (IETs) generated by a
  self-triggering rule is necessary for higher level planning and
  scheduling for control over shared or constrained resources as well
  as in the analytical quantification of the usage of
  communication or other resources compared to a time-triggered
controller. With these motivations, in this paper, we carry out a
systematic analysis of the evolution of IETs for linear
systems under region-based self-triggering rules (RBSTRs).
\vspace{-2ex}
\subsection{Literature review}

Event- and self-triggered control have been active areas of research
in the field of networked control systems~\cite{PT:2007, WH:2012,
  ML:2010, DT-SH:2017-book}. In the event-triggered control literature, typically the interest is only in showing the existence of
a positive lower bound on the IETs. Self-triggered
control~\cite{AA:2010} and periodic event triggered
control~\cite{WH:2013} guarantee a positive minimum IET
by design. However, in all these settings, a detailed analysis of the
IETs as a function of the state or time is typically
missing. 

Although it is not common, there are some works that analyze the
average of the IETs such
  as~\cite{KA-BB:2002,BD-AL-DQ:2017, FB:2017, PT-MF-JC:2018-tac,
    SB-PT:2021-cta}. References~\cite{PT:2016, QL:2017,
    JP-JPH-DL:2017, MJK-PT-JC-MF:2020-tac} provide necessary and
sufficient data rates for meeting the control goal with event-triggered control. On the other hand, \cite{BAK-DJA-WPMH:2018,
  FDB-DA-FA:2018} take a different approach and design event
triggering rules that ensure better performance than periodic control
for a given average sampling rate.

To the best of our knowledge, the first paper that studied the
  evolution of IETs generated by event-triggered control
  systems is~\cite{MV-PM-EB:2009}. This paper illustrates the periodic
  and chaotic patterns exhibited by the inter-event sequences of
  continuous time linear time invariant systems under homogeneous
  event-triggering rules. In the literature, it has been observed
that the IETs often settle to a steady state
value. Reference~\cite{RP-RS-WH-2022} seeks to explain this phenomenon
for planar linear systems with relative thresholding based
event-triggering rule, under a ``small'' thresholding parameter
scenario. It provides
sufficient conditions under which the IETs either
converge to some
neighborhood of a given constant or lie in some
neighborhood of a given constant or oscillate in a near periodic
manner. Reference~\cite{AK-MM:2018} proposes a method to characterize
  the sampling behavior of linear time-invariant event-triggered control systems by using finite-state abstractions of the
  system. References~\cite{GD-MM:2020} and \cite{GD-LL-MM:2021} extend
  this idea to nonlinear and stochastic event-triggered control
  systems, respectively. Similarly,~\cite{GAG-MM:2021} proposes an
  approach to estimate the smallest, over all initial states, average
  inter-sample time of a linear system under periodic event-triggered control through finite-state
  abstractions. Reference~\cite{GG-MM:2021} shows the robustness of
  the above approach to small enough model uncertainties. The recent
  paper~\cite{GG-MM:2022} analyzes the chaotic behavior of traffic
  patterns generated by periodic event-triggered control systems with
  the help of abstraction based methods. 
  
  Reference~\cite{GD-MM1:2021} designs self-triggering rules by studying
isochronous manifolds - sets of points in the state space with a given
IET. As the aim of this work is to design self-triggering
rules rather than to analyze IETs resulting from a given
triggering rule, the triggering rule is suitably modified to aid the
analysis. The recent paper~\cite{GA-KM-MM:2021} also proposes a self
triggered control scheme that provides near-maximal average
inter-sample time by using finite-state abstractions of a reference
event-triggered control, and by using ``early triggering''.

In~\cite{AR-PT:2020}, we provide a systematic way to analyze the
  evolution of IETs for planar linear systems under scale
  invariant event-triggering rules. In this work, we study the
  IET as a function of the angle of the state at an
  event, and then study the evolution of the angle of the state from
  one event to the next to indirectly understand the evolution of the
  IETs. Based on this method, we provide conditions for
  the convergence of IETs for the relative thresholding event-triggering rule.

\subsection{Contributions}

The major contribution of our work is that we provide a systematic way
to analyze the IETs, as a function of state or time, for
linear systems under region-based self-triggered control (RBSTC). The
  central idea behind our approach is that the next IET
  is a function of the state at the time of the event. Hence, studying
  the evolution of the state at event times indirectly informs us
  about the evolution of the IETs along the trajectories
  of the closed loop system. We provide quantitative results regarding the steady state behavior of
IETs and provide several necessary conditions and sufficient
conditions for the IETs to converge to a
constant or to a given periodic sequence.

References~\cite{AK-MM:2018, GD-MM:2020, GD-LL-MM:2021,
    GAG-MM:2021, GG-MM:2021, GG-MM:2022} characterize the sampling
  behavior of event-triggered control systems by using finite
  state-space abstractions of the system. However, this approach can
  be computationally very demanding. The
  references~\cite{GAG-MM:2021,GG-MM:2021,GG-MM:2022} analyze the
  periodic patterns exhibited by inter-sampling times of periodic event-triggered control systems, which is a special class of the
  RBSTC systems considered in our
  paper. \cite{GG-MM:2021,GG-MM:2022} provide a sufficient condition for the system to exhibit a given sequence of inter-sampling times under the assumption that the transformation matrix associated with the given inter-sampling time sequence is nonsingular. Additionally if the transformation matrix is mixed and of irrational rotations, then this condition is both necessary and sufficient. On the other hand, our results
  hold for a general class of systems and for arbitrary conic regions
  that are possibly even salient. Reference~\cite{GG-MM:2022} also analyzes the convergence of inter-sampling times to a given sequence and provides a necessary or sufficient condition for the same in some special cases. Compared
  to~\cite{GG-MM:2022} we also provide necessary and sufficient conditions for stability of positively
  invariant rays, subspaces and more general sub-regions all of which
  lead to convergence of IETs to a constant. Our results
  can also be adapted to study convergence of IETs to a
  given periodic pattern.

\vspace{-2.5ex}
\subsection{Notation}

Let $\real$, $\realz$, and $\realp$ denote the set of all real,
non-negative real and positive real numbers, respectively. For
  sets $\Mc{A}$ and $\Mc{B}$, $\Mc{A} \setminus \Mc{B}$ denotes
  $\Mc{A}$ set-minus $\Mc{B}$. Let $\nat$ and $\natz$ denote
the set of all positive and non-negative integers, respectively.
Let $\complex$ and $\complex^n$ denote the set of all complex
  numbers and the n-dimensional complex vector space,
  respectively. For any $x \in \complex^n$, $x^*$ and $x^H$ denote its
  conjugate and conjugate transpose, respectively, and let $\norm{x} \ldef \sqrt{x^Hx}$. For a square matrix $A \in \real^{n \times n}$, $\spec(A)$ denote
  	the spectrum of $A$ and $\rho(A)$ denote
  the spectral radius of
  $A$. $B_{\epsilon}(u)$ represents an $n-$dimensional ball of radius
$\epsilon$ centered at $u \in \complex^n$. For a set
  $M \subset \complex^n$,
  $B_{\epsilon}(M) \ldef \{x \in \complex^n: \text{inf}_{y \in
    M}\norm{x-y} \le \epsilon\}$. For a set $S \subset \real^n$, we let $\cl{S}$ denote the closure of
  $S$ in $\real^n$.
\vspace{-2ex}

\section{Problem Setup}\label{sec:problem-setup}

This section presents the dynamics of the system, the class
  of RBSTRs that we consider and the
objective of this paper.

\subsection{System Dynamics}
Consider a continuous-time, linear time invariant system,
\begin{subequations}\label{eq:system}
  \begin{equation}\label{eq:plant_dyn}
    \dot{x}(t) = Ax(t) + Bu(t),
  \end{equation}
  where \(x\in\real^n\) is the plant state and \(u\in\real^m\) is the
  control input, while $A \in \real^{n\times n}$ and
  $B \in \real^{n\times m}$ are the system matrices. Consider a
  sampled data controller and let \(\{t_k\}_{k\in \natz}\) be the
  sequence of event times at which the state is sampled and the
  control input is updated as follows,
  \begin{equation}\label{eq:control_input}
    u(t)=Kx(t_k), \quad \forall t\in[t_k, t_{k+1}).
  \end{equation}
\end{subequations}
For system~\eqref{eq:system}, we can write the solution \(x(t)\) as
\begin{equation}
  x(t)=G(\tau)x(t_k), \quad  \forall t \in [t_k,t_{k+1}), \label{eq:xt}
\end{equation}
where \(\tau \ldef t - t_k\) and
\begin{equation*}
G(\tau) \ldef \e^{A \tau} + \int_0^{\tau} \e^{A (\tau - s)}
BK \dd s.
\end{equation*}
In the literature, the control gain \(K\) is chosen such that
\(A_c \ldef A+BK\) is Hurwitz and the problem is typically to design a
rule that implicitly determines the sequence of event times
recursively. In this paper, we assume that the event times
\(\{t_k\}_{k\in \natz}\) are generated in a self-triggered manner.

\subsection{Region-Based Self-Triggering Rule}
In the RBSTC method, we partition the
state space into a finite number of conic regions
$R_i \subset \real^n$, for $i \in \{1,2,\ldots,r\}$. We then associate
each region with a fixed IET $\tau_i$.
An alternative way of partitioning the state-space is to first design
an event-triggering rule, which gives the IET as
$\tau_e(x)$.
  Suppose that the event-triggering rule ensures that there exists 
  $\tau_{\min}$, a positive lower bound on the IETs. This 
  is a common guarantee for many event-triggering rules. Similarly, if 
  there is an upper bound on the IETs generated by the 
  event-triggering rule, then we set it to $\tau_{\max}$. Otherwise, we 
  can choose $\tau_{\max}$ as any value such that $\tau_{\max}> 
  \tau_{\min} > 0$. Then, we  choose $\tau_i$'s such that 
  $\tau_1\leq \tau_{\min} <\tau_2<\ldots.<\tau_r \le 
  \tau_{\max} < \tau_{r+1}$.
Then we can partition the state space into $r$ regions as follows,
\begin{equation*}
  R_i \ldef \{cx \in \real^n : \tau_i \le \tau_e(x)<\tau_{i+1}, \ c \ge 
  0, \ x \neq 0\} \quad \forall i \in \{1,\ldots,r\}.
\end{equation*}
In either way, we make the following standing
  assumption.
\begin{enumerate}[resume, label=\textbf{(A\arabic*)}]
\item Each region $R_i$, $\forall i \in \{1, \ldots, r\}$, is a
    cone. $\tau_i \neq \tau_j$, $\forall i \neq j$ and
    $i, j \in \{1, \ldots, r\}$. Also, the intersection of null space
    of $G^l(\tau_i)$ and $R_i$ is $\{0\}$,
    $\forall l \in \{1,2,..,n\}$ and
    $\forall i \in \{1, \ldots, r\}$.  \label{A:R_i}
\end{enumerate}
Notice from~\eqref{eq:xt} that the solution $x(t)$ for each
  $t \in [t_k, t_{k+1})$ is a linear function of $x(t_k)$. Thus, it is
  reasonable to assume that each region $R_i$ is a cone. There is no
  loss of generality in the assumption that $\tau_i \neq \tau_j$ since
  if two different regions have the same $\tau$'s then they can be
  combined into a single region. The assumption that,
  $\forall l \in \{1,2,..,n\}$, the intersection of null space of
  $G^l(\tau_i)$ and $R_i$ is $\{0\}$, $\forall i \in \{1, \ldots, r\}$, is not restrictive as otherwise, it would
  mean that there are non-zero initial conditions for the state from
  which the state evolves to $0$ in finite time, under constant open
  loop control.

Now, we define the RBSTR by setting the
IET as
\begin{equation}\label{eq:STC}
t_{k+1}-t_k=\tau_i, \quad \text{if } x(t_k) \in R_i, \ i \in
\{1,\ldots,r\} .
\end{equation}
Thus, the region to which $x(t_k)$ belongs determines fully the
IET $t_{k+1} - t_k$, and as a result, the set of possible
IETs is finite. Hence, a region-based self-triggered
  controller is easier to implement compared to an event-triggered
  controller, specially if the regions are polytopes. Alternately, it
  suffices to check a triggering rule at a finite set of times as in
  periodic event-triggered control. This property of the IETs makes the analysis of their steady state behavior much easier
than in event-triggered control setting. As a result, it would be
  much easier to integrate RBSTC method
  with higher level planning and scheduling algorithms in the context
  of shared or constrained communication or computational resources.

  The popular periodic event-triggered control
  method could be thought of a special case of the RBSTC method. To the best of our knowledge, the
  idea of RBSTC was first proposed
  in~\cite{CF-etal:2012}, though the name was first introduced
  in~\cite{GD-MM1:2021}.
\subsection{Objective}

The main objective of this paper is to analyze the evolution of
  IETs along the trajectories of system \eqref{eq:system}
  for conic region-based self triggering rules~\eqref{eq:STC}.
  Moreover, we seek to provide analytical guarantees for the
  asymptotic behavior of IETs under these rules. 
  The
  approach we take is to analyze IET and the state at the
  next event as functions of the state at the time of the current
  event.

\section{Analysis of Evolution of Inter-Event Times} \label{sec:iet-self-trig}

In this section, we analyze the evolution of IETs along
the trajectories of system~\eqref{eq:system} under the
RBSTC method~\eqref{eq:STC}.  We carry
out this analysis through a map that describes
the normalized evolution of the state from one event to the
next. This is similar in spirit to our method in the
  event-triggered control setting for planar systems
  in~\cite{AR-PT:2020}. We define the \emph{normalized inter-event
    state jump map} or simply the \emph{gamma map} as
\begin{equation}
  \label{eq:gamma_map}
\tilde{x}(t_{k+1}) = \gamma(\tilde{x}(t_k)) \ldef \frac{G(\tau_i) \tilde{x}(t_k)}{\norm{G(\tau_i) \tilde{x}(t_k)}}, \quad i \text{
    s.t. } \tilde{x}(t_k) \in R_i,
\end{equation}
where $\tilde{x}(t_0) \ldef x(t_0)$. Thus,
$\tilde{x}(t_k) \ldef \frac{x(t_k)}{\norm{x(t_k)}}$ for all
$k \in \nat$.

We normalize the state at each iteration because the IETs
are determined solely by the ``direction'' of the state even if the
sequence $\{x(t_k)\}_{k \in \natz}$ converges to zero. Thus, to
understand the asymptotic behavior of the IETs, one must
know the direction in which the state $x(t_k)$ converges to zero, if
it does. We define a set $\bar{R} \subseteq R_i$ for some
$i \in \{1,\ldots,r\}$ as a \emph{positively invariant subregion (PIS)} if
$\bar{R}$ is positively invariant under the gamma
map~\eqref{eq:gamma_map}, that is $\gamma (x) \in \bar{R}$ for all
$x \in \bar{R}$.%

\begin{lem}\thmtitle{Necessary and sufficient condition for convergence 
of IETs} \label{rem:convergence_of_IET}%
	The IETs along the trajectories of 
	system~\eqref{eq:system} under the RBSTR~\eqref{eq:STC} converge to a steady state value for some
	initial condition if and only if there exists a PIS.
\end{lem}
\begin{proof}
	This result follows directly from the definition of a PIS and the assumption that $\tau_i \neq \tau_j$, for $\forall i \neq j$ and
	$i, j \in \{1, \ldots, r\}$.
\end{proof}
Lemma~\ref{rem:convergence_of_IET} establishes the connection
  between convergence of IETs and existence of PIS. So, we can analyze the steady state behavior
  	of the IETs along the trajectories of
  	system~\eqref{eq:system} under the RBSTR~\eqref{eq:STC}, by studying about the existence of a
  	PIS and by studying the stability of
  	such a subregion under the gamma map. Given Assumption~\ref{A:R_i}, it suffices to
  look for PISs that are cones. In fact,
  PISs are closely connected to the
  eigenspaces of the matrices $G(\tau_i)$. In order to present this
  idea in a unified manner irrespective of whether $G(\tau_i)$ has
  real or non-real eigenvalues, we first introduce some notation and a
  couple of definitions.
   
  For a vector $v \in \complex^n$, we define the $\rspantxt$ as
  \begin{equation*} 
    \rspan(v) := \mathrm{span} \{ v + v^*, \sqrt{-1} (v - v^*)
    \}, 
  \end{equation*}
  where the span is over the real numbers. Thus, if $v \in \real^n$
  then $\rspan(v)$ is a line and otherwise it is a plane. We are most
  interested in $\rspan(v)$ for $v$ eigenvectors of a matrix. We
  present this in the following definition.  
  
  \begin{definition}\thmtitle{$\reig$} \label{def:inv-space}
    For a matrix $M \in \real^{n \times n}$ with eigenvalue
    $\lambda \in \complex$, we say
    \begin{equation}
      \label{eq:inv-subspace}
      \invspace_\lambda(M) := \rspan(v), \ \mathrm{s.t.} \ Mv =
      \lambda v, \ v \neq 0,
    \end{equation}
    is an $\reig$ corresponding to eigenvalue $\lambda$. \remend
  \end{definition}
  Note that if $\lambda \in \real$, then $\invspace_\lambda(M)$ is a
  line in $\real^n$ and if $\lambda \in \complex \setminus \real$,
  then $\invspace_\lambda(M)$ is a plane in $\real^n$. Also, note that
  if $\lambda \in \complex \setminus \real$, our terminology ``$\reig$
  corresponding to $\lambda$'' is somewhat imprecise as
  $\invspace_\lambda(M)$ is really an invariant plane under the joint
  action of the complex conjugate eigenvalues $\lambda$ and
  $\lambda^*$. Finally, $\invspace_\lambda(M)$ may not be unique for
  each $\lambda$. If the geometric multiplicity of $\lambda$ is $p$,
  then the span of all possible $\invspace_\lambda(M)$ is a subspace
  of dimension $p$ and $2p$ if $\lambda \in \real$ and
  $\lambda \in \complex \setminus \real$, respectively. Next, we
  present a lemma on $\reig$s corresponding to an eigenvalue. This
  result helps us to provide a necessary condition for the existence
  of a PIS under the gamma
  map~\eqref{eq:gamma_map}.
  
\begin{lem}\thmtitle{Convergence to an
    $\reig$} \label{lem:convergence_to_invsubspace}
  Let $M \in \real^{n \times n}$ be a matrix with its spectrum
  $\spec(M)$ as $\{ \lambda \}$ or $\{ \lambda, \lambda^* \}$ for some
  $\lambda \in \real$ or $\lambda \in \complex \setminus \real$,
  respectively. In either case, suppose the geometric multiplicity of
  $\lambda$ and $\lambda^*$ is one. Then, for any
  $x \in \real^n \setminus \{0\}$, the sequence
  $\left\{ \frac{M^kx}{\norm{M^kx}} \right\}_{k \in \natz}$ converges
  to $\invspace_\lambda(M) \cap B_1(0)$, the intersection of the
  unique $\reig$ corresponding to $\lambda$ with the unit sphere.
\end{lem}
\begin{proof}
  We prove the result using the Jordan normal form $M_J$ of $M$. If
  $\lambda \in \real$ then $M_J$ contains a single Jordan block of
  size $n$ and if $\lambda \in \complex \setminus \real$ then $M_J$
  contains two decoupled Jordan blocks, each of size $n/2$ and
  corresponding to $\lambda$ and $\lambda^*$, respectively. In the
  latter case, one can obtain the Jordan form $M_J = V M V^{-1}$ by
  picking $V = [ V_1 \ \ V_1^*]$, where the columns of $V_1$ are
  linearly independent generalized eigenvectors corresponding to
  $\lambda$. This allows any $x \in \real^n$ to be expressed as
  $x = (V_1 + V_1^*)y + \sqrt{-1} (V_1 - V_1^*)z $ for
  $y, \ z \in \real^{(n/2)}$.

  As the Jordan blocks are decoupled, it suffices to consider an
  arbitrary Jordan block $J$ of size $p$, corresponding to eigenvalue
  $\lambda$. Note that $J^k$ is an upper triangle matrix with the form
  \begin{align*}
    J^k = \begin{bmatrix}
            \lambda^k & \binom{k}{1}\lambda^{k-1} & \binom{k}{2}\lambda^{k-2} & . & .&  \binom{k}{p-1}\lambda^{k-p+1}\\
            0 &	\lambda^k & \binom{k}{1}\lambda^{k-1} & . & .& \binom{k}{p-2}\lambda^{k-p+2} \\
            . &	. & . & . & .& .\\
            0 &	0 & 0 & . & .& \lambda^k\\
               \end{bmatrix},
  \end{align*}
  for $k \geq p$. Let $J^k_{i,j}$ be the element in $J^k$ at row $i$
  and column $j$. Then, observe that for all
  $j \geq i \in \{2, \ldots, p\}$,
  \begin{equation*}
    \lim_{k \to \infty }\frac{J^k_{i,j}}{J^k_{1,j}} = \lambda^{i-1}
    \frac{(j-1)! (k+1-j)!}{(j-i)! (k+i-j)!} = 0 .
  \end{equation*}
  Then, by invoking linearity, we can infer that
  $\forall x \in \complex^p \setminus \{0\}$, the sequence
  $\left\{ \frac{J^k x}{\norm{J^k x}} \right\}_{k \in \natz}$
  converges to the intersection of the eigenspace of $J$ with the unit
  sphere, which is $\{ %
  \pm \begin{bmatrix} 1 & 0 & \ldots & 0
  \end{bmatrix}\trsp
  \}$. From here, we can conclude that the claim in the result is
  true.
\end{proof}

With the help of Lemma~\ref{lem:convergence_to_invsubspace}, we now
present a necessary condition for the existence of a PIS and 
hence also for the possibility of convergence
of IETs to a constant.

\begin{prop}\thmtitle{Necessary condition for the existence of a
    PIS}\label{prop:existence_of_PIR}
  Consider the system~\eqref{eq:system} under the RBSTR~\eqref{eq:STC}. Suppose there exists a PIS $\bar{R} \subseteq R_i$, for some $i \in \{1,2,..,r\}$.
  Then,
  \begin{itemize}
  \item for each $\reig$, $\invspace_{\lambda}(G(\tau_i))$,
    for each eigenvalue $\lambda \in \spec( G(\tau_i) )$,
    $\invspace_{\lambda}(G(\tau_i)) \cap \bar{R}$ is positively
    invariant under the gamma map~\eqref{eq:gamma_map}.
  \item $\exists \mu \in \realp$ such that
    $S_{\mu}(G(\tau_i)) \cap \cl{\bar{R}} \cap B_1(0) \ne \emptyset$ where
   \begin{equation}\label{eq:gamma_invariant_subspace}
     S_{\mu}(G(\tau_i)) \ldef \mathrm{span}\{\invspace_{\lambda}(G(\tau_i)) :
     |\lambda|=\mu, \lambda \in \spec(G(\tau_i)) \}.
       \end{equation}
     \item for almost all $x \in \bar{R}$, the sequence
       $\{\gamma^k(x)\}_{k \in \natz}$ converges to the set
       $S_{\mu_{\max}}(G(\tau_i)) \cap \cl{\bar{R}} \cap B_1(0)$,
       where
       $\mu_{\max} \ldef \max\{\mu \in \realp : S_{\mu}(G(\tau_i))
       \cap \cl{\bar{R}} \cap B_1(0) \ne \emptyset \}$.
  \end{itemize}
\end{prop}
\begin{proof}
  Suppose there exists a PIS under the
  gamma map~\eqref{eq:gamma_map}, $\bar{R} \subseteq R_i$ for some
  $i \in \{1,2,..,r\}$. This means that for any $x \in \bar{R}$, the
  iterates of the gamma map are given by
  $\gamma^k(x) = \frac{G^k(\tau_i)x}{\norm{G^k(\tau_i)x}}$, for all
  $k \in \nat$ and for the fixed $\tau_i$ corresponding to
  $R_i$. Under the linear transformation $G(\tau_i)$,
  $\invspace_{\lambda}(G(\tau_i))$ for each
  $\lambda \in \spec(G(\tau_i))$ is positively invariant. So, the
  first claim is true. Note that, for each $\mu \in \realp$,
  $S_{\mu}(G(\tau_i))$ is also positively invariant under the linear
  transformation $G(\tau_i)$. Generalizing
  Lemma~\ref{lem:convergence_to_invsubspace}, we can say that for any
  $x \in \real^n$, the sequence $\{\gamma^k(x)\}_{k \in \natz}$
  converges to $S_{\mu}(G(\tau_i)) \cap B_1(0)$ for some
  $\mu \in \realp$. Hence, the second claim is true. To prove the
  final claim, we consider the Jordan normal form of $G(\tau_i)$. Note
  that almost all $x \in \bar{R}$ have a non-zero component along the
  subspace corresponding to at least one of the Jordan blocks
  corresponding to eigenvalues $\lambda$ with
  $|\lambda| = \mu_{\max}$. By a similar argument as in the proof of
  Lemma~\ref{lem:convergence_to_invsubspace}, we can again show that
  for all such initial $x$, the sequence
  $\{\gamma^k(x)\}_{k \in \natz}$ converges to the set
  $S_{\mu_{\max}}(G(\tau_i)) \cap \cl{\bar{R}}$.
\end{proof}

As
  $S_{\mu}( G(\tau_i) ) \cap \bar{R} \subseteq S_{\mu}(G(\tau_i)) \cap
  R_i$ for any $\mu \in \realp$,
  Proposition~\ref{prop:existence_of_PIR} helps in ruling out the
  existence of a PIS in each $R_i$. As we
  have a finite number of regions, we can determine the subspaces
  $S_{\mu}$, defined as in~\eqref{eq:gamma_invariant_subspace},
  corresponding to the $G$ matrix of each region. If one of these
  subspaces intersects with the closure of the corresponding region,
  then there is a possibility that a PIS
  exists. If none of the $S_{\mu}$ subspaces of $G(\tau_i)$, for each
  $i$, intersects with the closure of the corresponding region $R_i$,
  then it implies that there does not exist a PIS. Hence, in that case, the IETs do not
  converge to a steady-state value for any initial state of the
  system. Proposition~\ref{prop:existence_of_PIR} also suggests
  that it is sufficient to study the set
  $S_{\mu_{\max}}(G(\tau_i)) \cap \cl{\bar{R}}$ for any PIS $\bar{R} \subseteq R_i$, for some
  $i \in \{1,2,..,r\}$, as in practice, we would almost surely not
  observe convergence of trajectories to any other PISs. 

Next, we provide a necessary
  and sufficient condition for the existence of a PIS that is a subspace.

  \begin{prop}\thmtitle{Necessary and sufficient condition for the
      existence of a positively invariant subspace}
    Consider the system~\eqref{eq:system} under the RBSTR~\eqref{eq:STC}. Then there exists a positively invariant subspace $\bar{R} \subseteq R_i$ if and only if there is an
    $\reig$ $\invspace_{\lambda}(G(\tau_i))\subseteq R_i$
    for some eigenvalue $\lambda \in \spec(G(\tau_i))$, where
    $\invspace_{\lambda}(G(\tau_i))$ is as defined
    in~\eqref{eq:inv-subspace}.
  \end{prop}
	
  \begin{proof}
    By definition, $\invspace_{\lambda}(G(\tau_i))$ is positively
    invariant under the transformation $G(\tau_i)$. As
    $\invspace_{\lambda}(G(\tau_i))\subseteq R_i$,
    $\invspace_{\lambda}(G(\tau_i))$ is also a PIS under the gamma map~\eqref{eq:gamma_map}. Now, let us
    prove the converse of this statement. Let there exist a positively
    invariant subspace $\bar{R} \subseteq R_i$ for some
    $i \in \{1,\ldots,r\}$. Then $\bar{R}$ is a $G(\tau_i)$-invariant
    subspace. Thus, either $\bar{R}$ contains a real eigenvector of
    $G(\tau_i)$ corresponding to a real eigenvalue or $G(\tau_i)$ has
    complex conjugate eigenvalues with a compelx eigenvector $v$ that
    generates
    $\invspace_{\lambda}(G(\tau_i)) \subseteq \bar{R} \subseteq
    R_i$. This completes the proof of this result.
  \end{proof}
\vspace{-1.5ex}
Next we talk about a special class of PISs called positively invariant rays.
\begin{rem}\thmtitle{Positively invariant ray (PIR)} \label{rem:STC_fixed_ray}
  Consider the system~\eqref{eq:system} under the self-triggering
  rule~\eqref{eq:STC}. If $\exists x \in R_i$ such that
  $G(\tau_i)x=\alpha x$ for some $\alpha \in \realz$ and for some
  $i \in \{1,\ldots,r\}$, then $\{\beta x: \beta \ge 0\} \subseteq R_i$
  is a PIS, and we refer to it as a positively
  invariant ray (PIR). \remend
\end{rem}

It is easy to check for the existence of a PIR as
we have finite, to be specific $r$, number of regions. We can
determine the eigenvectors of $G(\tau_i)$ for each
$i \in \{1,\ldots,r\}$ and check if any of them belongs to the
corresponding region. Note that a PIS need
not always contain a PIR. Next, we provide a
necessary condition for the existence of a PIS that does not contain a PIR.

  \begin{prop}\thmtitle{Necessary condition for the existence of a
      PIS that does not contain a
      PIR}\label{prop:existence_of_PIRegion_without_PIRay}
    Consider the system~\eqref{eq:system} under the RBSTR~\eqref{eq:STC}. There exists a PIS    $\bar{R} \subseteq R_i$, for some $i \in \{1,\ldots,r\}$ whose
    closure does not contain a PIR only if either
    one of the following conditions holds.
    \begin{itemize}
    \item $\exists$ an eigenvector $v$ of $G(\tau_i)$ corresponding to
      a real negative eigenvalue $\lambda$ such that
      $\invspace_{\lambda}(G(\tau_i)) \subseteq \bar{R}$.
    \item $G(\tau_i)$ has two distinct eigenvalues with same magnitude
      $\mu$ such that
      $S_{\mu}(G(\tau_i)) \cap \cl{\bar{R}} \ne \{0\}$.
    \end{itemize}
  \end{prop}
  \begin{proof}
    Suppose neither of the two given conditions are satisfied. Then according to
    Proposition~\ref{prop:existence_of_PIR}, $\cl{\bar{R}}$ should
    contain an eigenvector of $G(\tau_i)$ corresponding to a real positive eigenvalue. So, the claim is true.
  \end{proof}

  Given the results and observations so far, we can give a
  non-exhaustive classification of the different classes of PISs that are possible in general.

\begin{rem}\thmtitle{Classification of PISs} \label{rem:classes_of_PI_subregions}
  A non-exhaustive list of possible PISs,
  $\bar{R}$ that are subsets of $R_i$ for some
  $i \in \{1, \ldots, r\}$, under the gamma map~\eqref{eq:gamma_map}
  is as follows. In the list, we also provide the conditions under
  which each case would occur.
  \begin{itemize}
  \item $\bar{R}$ a ray that is in the span of an eigenvector of
    $G(\tau_i)$ corresponding to a positive eigenvalue.
  \item $\bar{R}$ a union of finitely many rays, which are in the span
    of eigenvectors of $G(\tau_i)$ corresponding to eigenvalues
    $\lambda > 0$ and eigenvalues $- \lambda < 0$.
  \item $\bar{R}$ a union of finitely many rays, which are in the span
    of $\invspace_{\lambda}(G(\tau_i))$ for an eigenvalue
    $\lambda \in \complex \setminus \real$ and $\text{Arg}(\lambda)$
    is a rational multiple of $\pi$.
  \item $\bar{R}$ a plane of the form $\invspace_{\lambda}(G(\tau_i))$
    for an eigenvalue $\lambda \in \complex \setminus \real$.
  \item $\bar{R}$ that is a subspace spanned by a set of real
    generalized eigenvectors or pairs of complex conjugate generalized
    eigenvectors of the $G(\tau_i)$.
  \item $\bar{R}$ that is the span or more generally the cone
    generated by PISs that belong to the
    classes mentioned above.
  \end{itemize}
\end{rem}

As we see from this list, there are many possibilities for the
PISs in a region $R_i$. Giving unified
necessary and sufficient conditions for their existence would be
cumbersome and involve arguments and ideas that are repetitive. At the
same time, handling each case separately is relatively straight
forward. So, we skip presenting further analysis of the specific cases
for the sake of brevity.
\vspace{-2ex}
\subsection{Stability analysis of positively invariant subregions}

In this subsection, we analyze stability of a PIS. Notice that if the state $x(t_k)$ converges to
  $0$ then it converges to every subspace. However, we are really
  interested in the direction in which the state evolves whether it
  converges to zero or not since the IETs are determined
  by only the ``direction'' of the state. Hence, we employ the gamma
  map~\eqref{eq:gamma_map} for the stability analysis. Notice that
  given Assumption~\ref{A:R_i}, $\gamma(x) \in B_1(0)$ for any
  $x \in \real^n \setminus \{0\}$. Thus, given a PIS $\mathscr{M} \subset R_i$, for some
  $i \in \{1, \ldots, r\}$, we study stability of the set
  $\mathscr{M} \cap B_1(0)$ under the gamma map. Considering the
  variety of PISs that may exist, as seen
  in Remark~\ref{rem:classes_of_PI_subregions}, we focus the stability
  analysis to PISs $\mathscr{M}$ that are
  the intersection of $R_i$ and the generalized eigenspace of
  $G(\tau_i)$ corresponding to an eigenvalue
  $\lambda \in \spec(G(\tau_i))$ for some $i \in \{1, \ldots, r\}$. We
  make this choice as such sub-regions are fundamental and we can
  still present reasonably unified results that one may easily
  generalize. Next, we present a
  result which provides a necessary and sufficient condition for
  stability, asymptotic stability and instability of the intersection
  of such a PIS with the unit sphere.

  \begin{thm}\thmtitle{Necessary and sufficient condition for the
      intersection of a PIS with the unit
      sphere to be stable} \label{thm:asy_stable_ray}
    Consider system~\eqref{eq:system} under the RBSTR~\eqref{eq:STC}. Suppose the cone $R_i$, for some
    $i \in \{1,2,..,r\}$, is solid and let there exist a closed
    PIS $\bar{R}$ such that
    $\bar{R}\setminus\{0\}$ is in the interior of $R_i$. For each
    $\lambda \in \spec(G(\tau_i))$, let
    $\hat{\mathscr{M}} := \mathscr{M} \cap B_1(0)$, where
    \begin{equation*}
      \mathscr{M} := \bar{R} \cap
      \mathrm{span}\{\rspan(v):(G(\tau_i)-\lambda I)^lv=0, \ l \in \nat\} .
    \end{equation*}
    $\mathscr{M}$ is a PIS in $R_i$ or
    equivalently $\hat{\mathscr{M}}$ is positively invariant under the
    gamma map. Further if, $\mathscr{M} \ne \{0\}$ or equivalently
    $\hat{\mathscr{M}} \neq \emptyset$ then the following hold.
    
    \begin{itemize}
  \item If $\lambda$ is non-defective then, under the gamma map~\eqref{eq:gamma_map},
  $\hat{\mathscr{M}}$ is
  stable if and only if $|\lambda|=\rho(G(\tau_i))$ and all $q \in \spec(G(\tau_i))\setminus\{\lambda, \lambda^{\star}\}$, such that $|q|=|\lambda|$, are non-defective.
  
\item If $\lambda$ is defective then $\hat{\mathscr{M}}$ is stable if
  and only if $|\lambda|=\rho(G(\tau_i))$,
  $\mathrm{span}\{\rspan(v):(G(\tau_i)-\lambda I)^lv=0, \ l \in \nat\}
  \subseteq \mathscr{M}$ and $|q| < |\lambda|$ for all
  $q \in \spec(G(\tau_i)) \setminus \{\lambda, \lambda^{\star}\}$.
  
\item $\hat{\mathscr{M}}$ is asymptotically stable if and only if it
  is stable, $|q| < |\lambda|$ for all
  $q \in \spec(G(\tau_i)) \setminus \{\lambda, \lambda^{\star}\}$, and
  one of the following conditions holds: (a)
  $\rspan(v) \subseteq \mathscr{M}$, for every eigenvector $v$ of
  $G(\tau_i)$ corresponding to $\lambda$, or (b) $\lambda \in \realp$
  is an eigenvalue with algebraic multiplicity equal to one.
   
  \item $\hat{\mathscr{M}}$ unstable if and only if it is not stable.
  \end{itemize}
		
\end{thm}
\begin{proof}
  The $\real$-generalized eigenspace corresponding to each
  $\lambda \in \spec(G(\tau_i))$, i.e.
  $\mathrm{span}\{\rspan(v):(G(\tau_i)-\lambda I)^lv=0, \ l \in
  \nat\}$, is positively invariant under the linear map
  $G(\tau_i)$. Thus, $\mathscr{M}$ and $\hat{ \mathscr{M} }$ are
  positively invariant under the gamma map.

  In order to prove the rest of the claims, we first make some general
  observations and setup some notation. Note that for any $x \in R_i$
  such that $\gamma^j(x) \in R_i, \ \forall j \in \natz \cap [0, k]$,
  \begin{equation*}
    \gamma^k(x) = \frac{ G^k(\tau_i) x }{ \norm{G^k(\tau_i) x} } =
    \frac{ J^k x }{ \norm{J^k x} }, \quad J := \frac{ G(\tau_i) }{|
      \lambda |}.
  \end{equation*}
  Consider the partition $\{q_i\}_{i=1}^4$ of all eigenvalues
  of $J$, where
  \begin{align*}
    &q_1 := \{ q \in \sigma(J) : |q| > 1 \}, \ 
      q_2 := \{ \lambda / |\lambda| , \ \lambda^* / |\lambda| \},
    \\
    &q_3 := \{ q \in \sigma(J) \setminus q_2 : |q| = 1 \}, \
      q_4 := \{ q \in \sigma(J) : |q| < 1 \} .
  \end{align*}
  We assume without loss of generality that $G(\tau_i)$ and hence $J$
  is in the real Jordan form. Let $J := \text{diag}(J_1,J_2,J_3,J_4)$,
  where each $J_i$ is the block diagonal matrix of all Jordan blocks
  of $J$ corresponding to the eigenvalues in $q_i$. Let %
  $
  \begin{bmatrix}%
    E_1 & E_2 & E_3 & E_4
  \end{bmatrix} := I_n
  $, %
  where $I_n$ is the $n \times n$ identity matrix and for each
  $i$, $E_i$ has the same number of columns as $J_i$.
   
  Now, we prove the sufficiency for stability of
  $\hat{ \mathscr{M} }$. As $\hat{\mathscr{M}}$ is in the interior of
  $R_i$, we can find a $\rho >0$ such that
  $B_{\rho}(\hat{\mathscr{M}}) \subset R_i$. Let
  $z \ldef \sum_{j=1}^{4} E_jz_j \in B_{\delta}(\hat{\mathscr{M}})$,
  for some $\delta \in (0, 1)$, and let
  $\hat{z} \ldef E_2\hat{z}_2 \in \displaystyle \argmin_{y \in \hat{
      \mathscr{M} }}\norm{z-y}$, which implies that
  $\norm{z-\hat{z}}\le\delta$. Note that, for any $k \in \natz$,
  \begin{equation}\label{eq:stability_inequality}
  \begin{aligned}
 \norm{\frac{J^kz}{\norm{J^kz}}-\frac{J^k\hat{z}}{\norm{J^k\hat{z}}}}&=\norm{\frac{J^kz}{\norm{J^kz}}-\frac{J^kz}{\norm{J^k\hat{z}}}+\frac{J^kz}{\norm{J^k\hat{z}}}-\frac{J^k\hat{z}}{\norm{J^k\hat{z}}}}, \\
 & \le \frac{|\norm{J^k\hat{z}}-\norm{J^kz}|+\norm{J^kz-J^k\hat{z}}}{\norm{J^k\hat{z}}},\\
 & \le \frac{2\norm{J^k(z-\hat{z})}}{\norm{J^k\hat{z}}}.
 \end{aligned}
  \end{equation}

 Under the given conditions for stability, $J_1$ is non-existent. As
 $J_4$ is Schur stable, there exists a positive definite matrix $P$
 such that $J_4^T P J_4 - P \le 0$. If $\lambda$ is defective, under
 the given conditions for stability, $J_3$ is also non-existent. This
 implies that, for any $k \in \natz$,
  \begin{align*}
    \norm{\frac{J^kz}{\norm{J^kz}}-\frac{J^k\hat{z}}{\norm{J^k\hat{z}}}}
    & \le \frac{2 \left(\norm{J_2^k(z_2-\hat{z}_2)}+\norm{J_4^kz_4}
      \right)}{\norm{J_2^k\hat{z}_2}} \leq c_1\delta,
  \end{align*}
  for some $c_1 > 0$ that is independent of $k$ and $\delta$. The last
  inequality follows from the fact that
  $z_2 = r \frac{\hat{z}_2}{\norm{\hat{z}_2}}$ for some
  $r \in (0, \delta)$ and
  $\frac{\norm{J_4^kz_4}}{\norm{J_2^k\hat{z}_2}}$ is upper bounded by
  some positive real number $\bar{c}\delta$ as
  $\norm{\hat{z}} = \norm{\hat{z}_2} = 1$, $\norm{J_4^kz_4}$ decreases
  exponentially while the smallest singular value of $J_2^k$ can
  decrease only at a polynomial rate, as shown in Theorem~2.1
  in~\cite{IG-etal:1996}. If $\lambda$ is non-defective, then $J_2$
  and $J_3$ are orthogonal matrices. Hence, in either case, we can say
  that there exists a $c > 1$ that is independent of $\delta$ and $k$
  such that
  \begin{equation*}
   \norm{\frac{J^kz}{\norm{J^kz}}-\frac{J^k\hat{z}}{\norm{J^k\hat{z}}}}
   \le c \delta, \quad \forall z \in  B_{\delta}(\hat{\mathscr{M}}), \
   \forall k \in \nat.
 \end{equation*}
 As $\frac{J^k\hat{z}}{\norm{J^k\hat{z}}} \in \hat{\mathscr{M}}$, we
 can also say that
 $\frac{J^kz}{\norm{J^kz}} \in B_{c \delta}(\hat{\mathscr{M}})$. Thus,
 given any $\epsilon > 0$, we can choose
 $\delta < \min\{\epsilon, \rho, 1\}/c >0$ such that
 $z \in B_{\delta}(\hat{\mathscr{M}})$ implies
 $\frac{J^kz}{\norm{J^kz}} \in B_{\epsilon}(\hat{\mathscr{M}}) \cap
 B_{\rho}(\hat{\mathscr{M}}) \subset R_i$ , $\forall k \in
 \nat$. Thus,
 $\gamma^k(z) = \frac{J^kz}{\norm{J^kz}} \in
 B_{\epsilon}(\hat{\mathscr{M}}), \ \forall k \in \nat$, which then
 means that $\hat{\mathscr{M}}$ is stable. Note that the choice
 $\delta < 1/c$, ensures that $B_{c \delta}(\hat{\mathscr{M}})$
 excludes the zero vector and hence
 $B_{c \delta}(\hat{\mathscr{M}}) \subset R_i$.
	
 Now, we prove the sufficiency for asymptotic stability. From the
 sufficiency for stability, we know that if
 $z \in B_{\delta}(\hat{\mathscr{M}})$ for small enough $\delta > 0$
 then $\gamma^k(z) = \frac{J^kz}{\norm{J^kz}} \in R_i$ and
 $\norm{ \gamma^k(z) } = 1, \ \forall k \in \nat$. Under the given
 conditions, note that $J_1$ and $J_3$ are non-existent and since
 $J_4$ is Schur stable, $J_4^k z_4$ converges to zero as
 $k \rightarrow \infty$. Thus, we can say that $J^kz$ asymptotically
 converges to $\mathrm{span}\{E_2\} = \mathrm{span}\{\mathscr{M}\}$. But
 $\mathscr{M} = \mathrm{span}\{\mathscr{M}\}$ in the subcase (a). In
 subcase (b), if $\mathscr{M} \neq \{0\}$ then $\mathscr{M}$ is either
 the eigenspace corresponding to $\lambda$, which is a line, or one of
 the two rays contained within it. In either case, $J_2^k z_2 = z_2$
 for all $k \in \natz$ as $q_2 = \{1\}$ and $E_2 z_2 \in
 \mathscr{M}$. This proves asymptotic stability of $\hat{\mathscr{M}}$
 under the gamma map.

 Now, we prove the necessity for asymptotic stability. If neither of
 the sub-cases (a) or (b) holds, then we can always choose a
 $z \in B_{\delta}(\hat{\mathscr{M}})$ such that
 $z \in \mathrm{span}\{ \mathscr{M} \}$ but $z \notin
 \mathscr{M}$. For such a $z$, the distance from $\gamma^k(z)$ to
 $\hat{\mathscr{M}}$ is the same for all $k \in \nat$ and equals the
 distance from $z$ to $\hat{\mathscr{M}}$. If $q_3$ is non-empty, we
 can again choose a $z$ with a non-zero $z_3$ and hence $J^k z$ does
 not converge to $\mathrm{span}\{ \mathscr{M}\}$. If $q_1$ is
 non-empty then $\hat{\mathscr{M}}$ is not even stable. Thus, in each
 of these sub-cases $\hat{\mathscr{M}}$ is not asymptotically stable.

 Finally, we prove the necessity for stability or equivalently, the
 sufficiency for instability. If $q_1$ is non-empty then by
 considering $\displaystyle \frac{ G(\tau_i) }{ \rho(G(\tau_i)) }$
 instead of $J$, we can again say that there exist initial
 $z \in B_{\delta}(\hat{\mathscr{M}})$ for which $J^k z$ converges to
 the span of columns of $E_1$, which means $\hat{\mathscr{M}}$ is not
 stable or asymptotically stable. If $\lambda$ is non-defective and
 there is a defective eigenvalue in $q_3$, then we can choose a
 $z \in B_{\delta}(\hat{\mathscr{M}})$ such that $\norm{J_3^k z_3}$
 grows arbitrarily large while $\norm{J_2^k z_2}$ remains constant. If
 $\lambda$ is defective and $q_3$ is non-empty then, again we can
 choose a $z \in B_{\delta}(\hat{\mathscr{M}})$ such that
 $\norm{J_2^k z_2}$ is arbitrarily small for \emph{some} $k \in \nat$
 while $\norm{J_3^k z_3}$ remains a constant for all $k \in
 \natz$. This proves that the conditions in the result are also
 necessary for stability or sufficient for instability of
 $\hat{\mathscr{M}}$ under the gamma map.
\end{proof}

Theorem~\ref{thm:asy_stable_ray} implies that stability of a
  PIS can be determined by analyzing the
  spectrum and the generalized eigenspaces of the $G$ matrix
  corresponding to the region. Note also that if $\lambda$ is
  defective then conditions for stability also imply asymptotic
  stability. We can also relax some of the assumptions or generalize
  some of the claims in Theorem~\ref{thm:asy_stable_ray} but we skip
  discussing them due to space constraints.

\vspace{-2ex}
\subsection{Convergence of IETs to a periodic sequence}

  In this subsection, we provide conditions for the convergence of
  IETs to a given periodic sequence.
  
   	 \begin{lem}\thmtitle{Necessary and sufficient condition for the
  			existence of a given periodic sequence of IETs}
  		\label{rem:convergence_to_periodic_sequence}
  		Let $\tau_P \ldef \tau_{j_1}\tau_{j_2}..\tau_{j_p}$ represent a
  		finite sequence of IETs and let $\tau_P^{\omega}$ be
  		the sequence obtained by repeating $\tau_P$ infinitely. The
  		periodic sequence $\tau_P$ of IETs along the
  		trajectories of the system~\eqref{eq:system} under the
  		RBSTR~\eqref{eq:STC} \emph{exists} if
  		and only if there exists a subregion
  		$\bar{R}_P \subseteq R_{\tau_P}$, which is positively 
  		invariant under the map
  		\begin{equation}\label{eq:gamma_tau_P}
  		\gamma_{\tau_P}(x) \ldef
  		\frac{G_{\tau_P}x}{\norm{G_{\tau_P}x}}, \ G_{\tau_P} \ldef
  		G(\tau_{j_{p}})G(\tau_{j_{p-1}})..G(\tau_{j_2})G(\tau_{j_1}),
  		\end{equation} 
  		\begin{equation}\label{eq:R_tau_P}
  		\begin{aligned}
  		R_{\tau_P} \ldef \{x \in \real^n: &\  x \in R_{j_1},
  		G(\tau_{j_1})x \in R_{j_2}, ..., \\ & 
  		G(\tau_{j_{p-1}})..G(\tau_{j_2})G(\tau_{j_1})x \in R_{j_p}\}.
  		\end{aligned}
  		\end{equation}
  	\end{lem}
  \begin{proof}
  	This result is a direct extension of 
  	Lemma~\ref{rem:convergence_of_IET}.
  \end{proof}

   \begin{cor}\thmtitle{Necessary condition for the convergence of
       IETs to a given periodic
       sequence}\label{cor:convergence_to_periodic_sequence}
     Consider system~\eqref{eq:system} under the RBSTR~\eqref{eq:STC}. Let $\tau_P$ and $\tau_P^{\omega}$ be
     defined as in
     Lemma~\ref{rem:convergence_to_periodic_sequence}. Suppose there
     exists a subregion $\bar{R}_P \subseteq R_{\tau_P}$ which is
     positively invariant under the $\gamma_{\tau_P}(.)$ map
     in~\eqref{eq:gamma_tau_P}. Then,
     \begin{itemize}
     \item there exists $\mu \in \realp$ such that
       $S_{\mu}(G_{\tau_P}) \cap \cl{\bar{R}_P} \ne \{0\}$ where
       $G_{\tau_P}$ is defined in~\eqref{eq:gamma_tau_P} and
       $S_{\mu}(.)$ is defined in~\eqref{eq:gamma_invariant_subspace}.
     	\item for almost all $x \in \bar{R}_P$, the sequence
     	$\{\gamma^k_{\tau_P}(x)\}_{k \in \natz}$ converges to the set
     	$S_{\mu_{\max}}(G_{\tau_P}) \cap \cl{\bar{R}_P}$, where
     	$\mu_{\max} \ldef \max\{\mu \in \realp :
     	S_{\mu}(G_{\tau_P}) \cap \cl{\bar{R}_P} \ne \{0\}\}$.
     \end{itemize}
     
   \end{cor}
   \begin{proof}
     This result is a direct extension of
     Proposition~\ref{prop:existence_of_PIR}.
   \end{proof}

   Similarly, all the other results in this section related to
   existence, stability and asymptotic stability of positively
   invariant subregions hold exactly by replacing $R_i$, $G(\tau_i)$
   and $\gamma(.)$ with $ R_{\tau_P}$, $G_{\tau_P}$ and
   $\gamma_{\tau_P}(.)$, respectively. In other words, we have
   necessary and sufficient conditions for local convergence of
   IETs to a given periodic sequence.

   Note that, references~\cite{GAG-MM:2021,GG-MM:2021,GG-MM:2022} also
   do similar analysis of the periodic patterns exhibited by
   IETs. But, their analysis is based on the assumption
   that the matrix which transfers the state from one sampling time to
   the next is mixed and of irrational rotations. On the other hand,
   our results hold for general $G$ matrices. Compared
   to~\cite{GAG-MM:2021,GG-MM:2021}, we also have necessary and
   sufficient conditions for stability and asymptotic stability of the
   PISs.

\vspace{-2ex}
\section{Numerical examples} \label{sec:numerical-examples}

In this section, we illustrate our results through two numerical
examples and simulations.

\subsubsection*{Example 1}
Let us consider a 3-dimensional system,
\begin{equation*}
\dot{x}=\begin{bmatrix}
0 & 1 & 0\\ 0 & 0 & 1\\ -6 & 7 & 0
\end{bmatrix}x+\begin{bmatrix} 0 \\ 0 \\ 1
\end{bmatrix}u \rdef Ax + Bu .
\end{equation*}
$A$ has real eigenvalues at $\{1,2,-3\}$. The control gain
$K=[0 \quad -18 \quad -6]$ ensures that $A_c$ has real eigenvalues at
$\{-1,-2,-3\}$. 
{\setlength\intextsep{0pt}
	\begin{figure}[h]
		\centering
		\begin{subfigure}{4cm}
			\includegraphics[width=3.8cm]{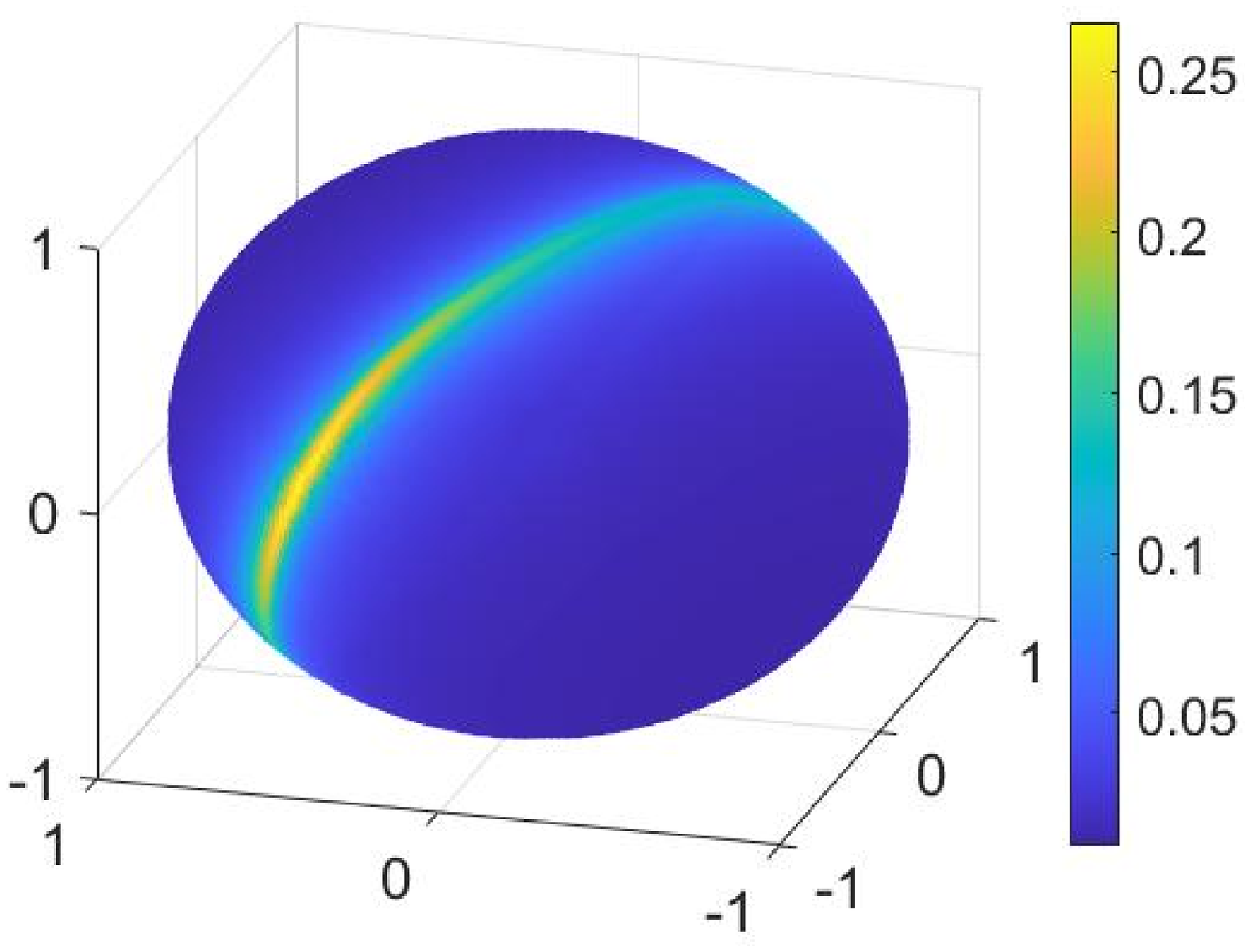}
			\caption{Inter-event time function}
			\label{fig:case7_tau_e}
		\end{subfigure}
		\quad
		\begin{subfigure}{4cm}
			\includegraphics[width=3.8cm]{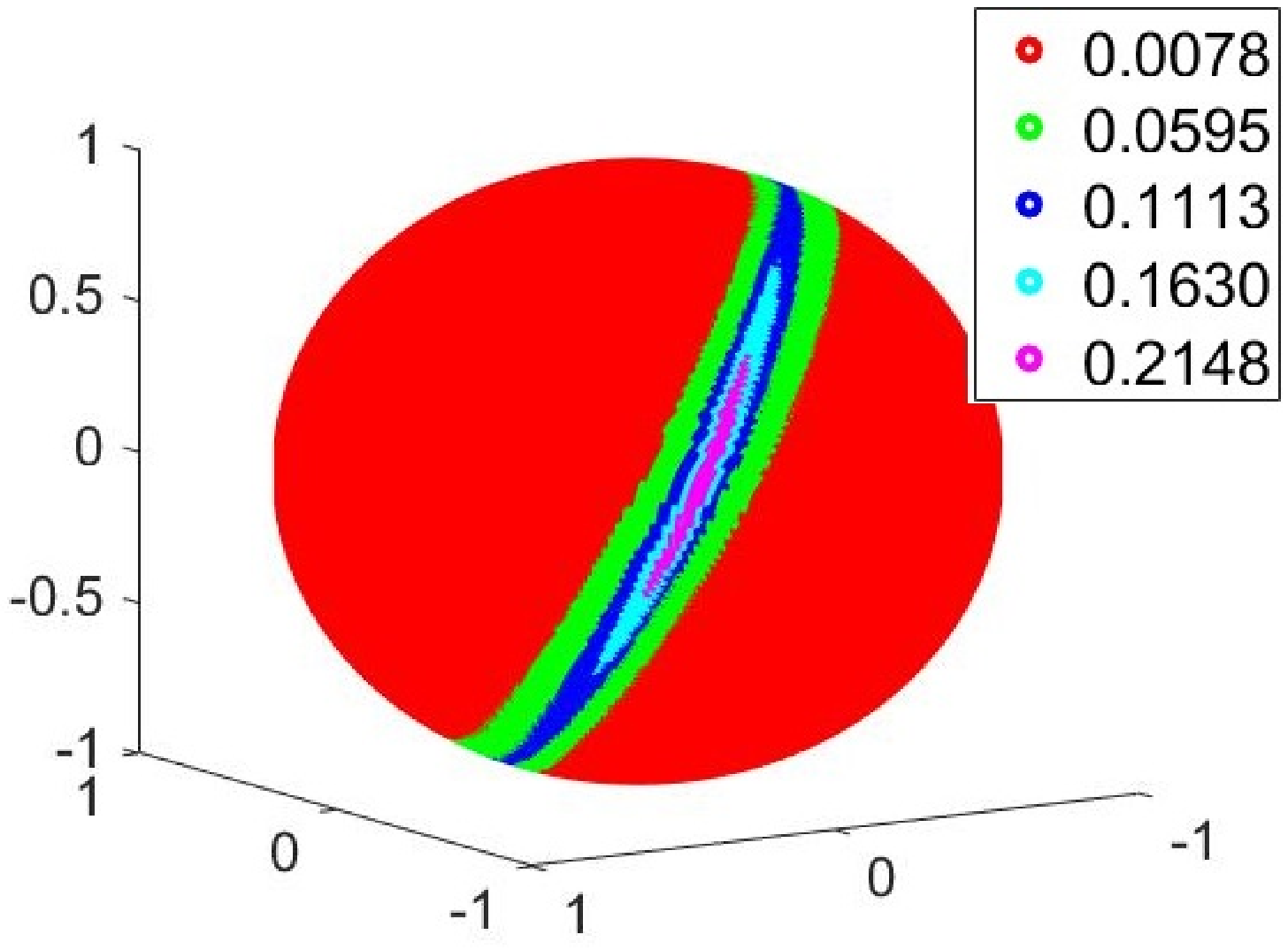}
			\caption{State-space partitions}
			\label{fig:case7_ssp}
		\end{subfigure}
		\quad
		\begin{subfigure}{4cm}
			\includegraphics[width=3.8cm]{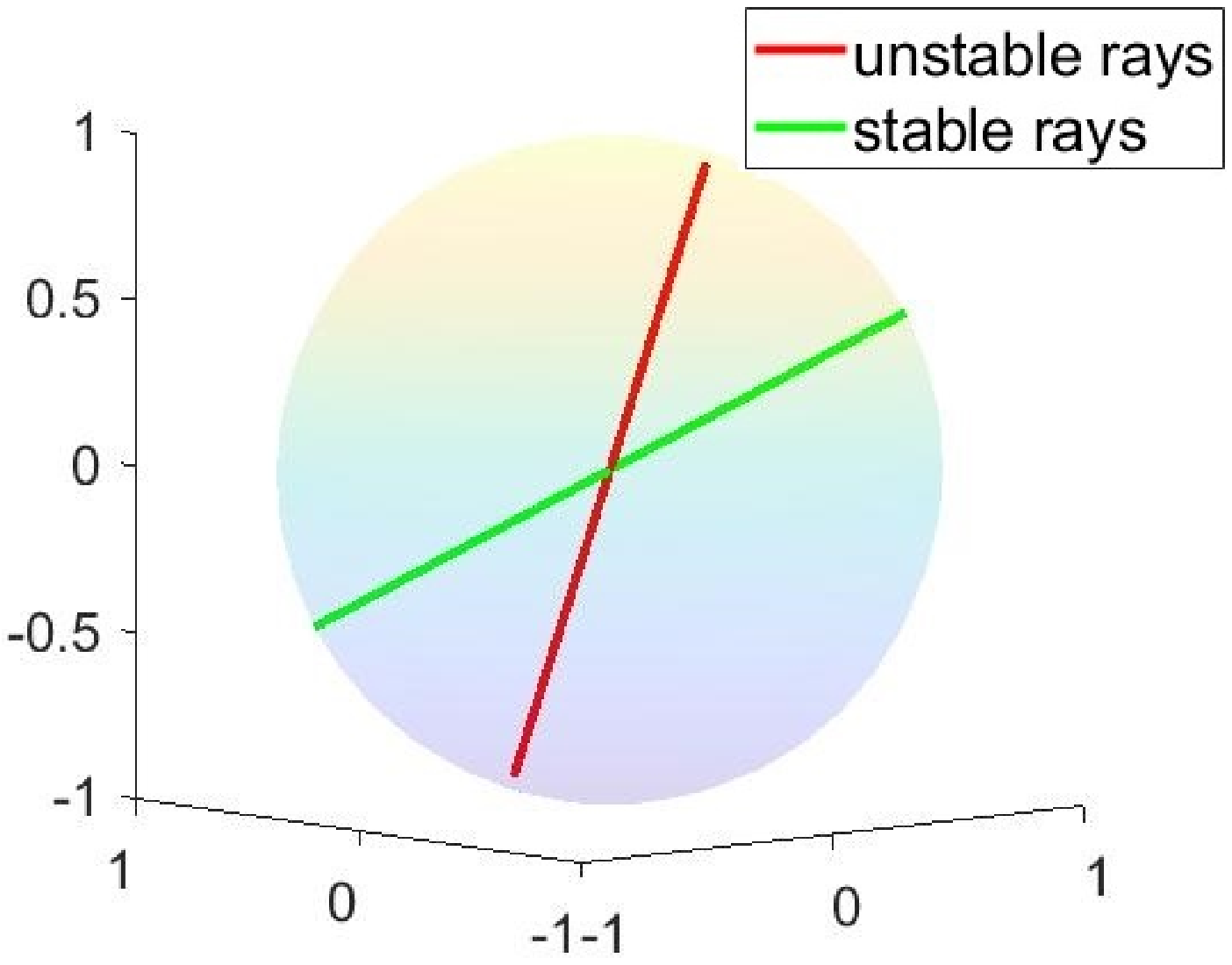}
			\caption{Positively invariant rays}
			\label{fig:case7_PIRs}
		\end{subfigure}
		\quad
		\begin{subfigure}{4cm}
			\includegraphics[width=3.8cm]{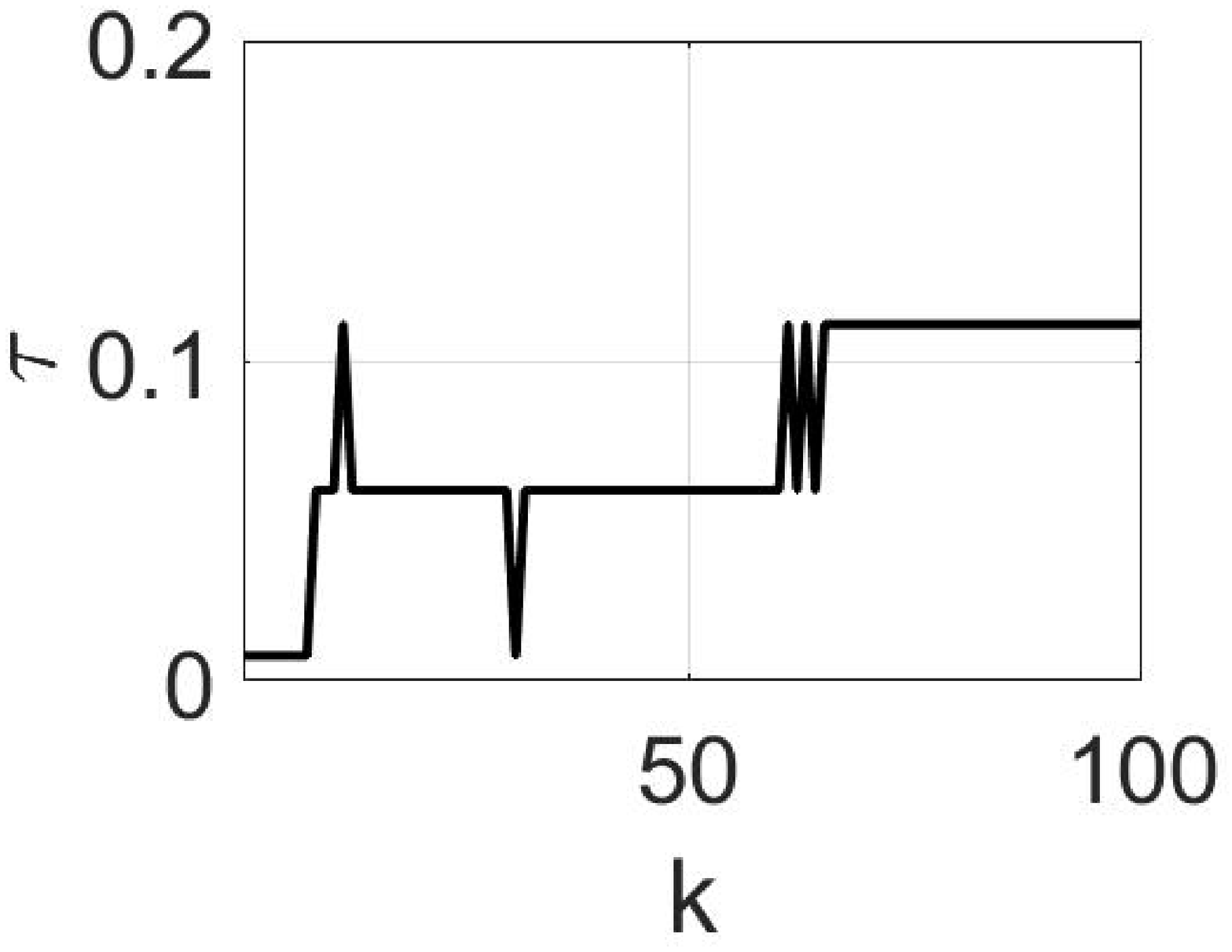}
			\caption{Inter-event times evolution}
			\label{fig:case7_tau_k}
		\end{subfigure}
		\caption{Simulation results of Example 1.}
		\label{fig:3-dim_system}
	\end{figure}
}

Figure~\ref{fig:3-dim_system} presents the simulation
results of Example 1. Figure~\ref{fig:case7_tau_e} shows the
value of the IET function $\tau_e(.)$ on the unit sphere
under the relative thresholding event-triggering rule. The global 
minimum and global maximum of $\tau_e(.)$, respectively, are
$\tmin=0.0088$ and $\tmax=0.2655$. We choose the number of regions
$r=5$ and discretize the interval $[\tau_{\min},\tau_{\max}]$ into $5$
equal parts. Then, we partition the state space into $5$ regions and
associate each region with the corresponding IET. Figure~\ref{fig:case7_ssp} shows the state-space partitions and
the corresponding IET. In this case, the $G$
  matrices corresponding to all the regions have real distinct
  positive eigenvalues. So, we can eliminate the possibility of
  existence of a PIS that does not contain
  a PIR. We can find numerically the four
PIRs, which are shown in
Figure~\ref{fig:case7_PIRs}. We can also analyze the stability of
these PIRs by analyzing the spectrum of the
corresponding $G$ matrices. Note that, here two of the PIRs are asymptotically stable and the other two are
unstable. The two asymptotically stable PIRs
  form a positively invariant line passing through the origin and the
  point $(0.6813,-0.5616,0.4695)$. Similarly, the two unstable
  PIRs form a positively invariant line passing
  through the origin and the point $(0.1048,-0.3145,0.9435)$. Next,
we choose an arbitrary vector close to one of the asymptotically
stable PIRs as the initial state and
Figure~\ref{fig:case7_tau_k} presents the evolution of inter-event
times under the RBSTR~\eqref{eq:STC}. We
can see that the IETs converge to a steady state value
which is exactly equal to the IET corresponding to one of
the asymptotically stable PIRs.

\subsubsection*{Example 2}
Next, let us consider a $5^{\text{th}}$ order system,
	\begin{equation*}
	\dot{x}=\begin{bmatrix}
	0 & 1 & 0 & 0 & 0\\ 0 & 0 & 1 & 0 & 0\\	0 & 0 & 0 & 1 & 0\\ 0 & 0 & 0 & 0 & 1\\	30 & -79 & 80 & -40 & 10
	\end{bmatrix}x + \begin{bmatrix}
	0\\0\\0\\0\\1
	\end{bmatrix}u \rdef Ax+Bu.
	\end{equation*}
	The system matrix $A$ has eigenvalues at $\{1, 2, 3, 2\pm i\}$. We choose the
	control gain $K$ so that $A_c \ldef A+BK$ has eigenvalues at $\{-0.1,-0.15,-0.2,-0.25,-0.3\}$. We first partition the state 
	space into $50$ solid convex cones by using the recursive zonal equal area 
	sphere partitioning toolbox~\cite{PL:2005,PL:2006} in MATLAB which 
	helps to partition higher dimensional unit spheres into regions of 
	equal Lebesgue measure. Note that none of our results require the 
	region-based self-triggered controller guarantee asymptotic stability 
	of the equilibrium at the origin of the closed loop system. So, we 
	arbitrarily assign an IET from the interval 
	$[0.03,0.23]$ for each convex cone. As per Assumption~\ref{A:R_i}, if 
	there are two convex cones with the same IET then they 
	belong to the same region.
	
	{\setlength\intextsep{0pt}
		\begin{figure}[h]
			\centering
			\includegraphics[width=4cm]{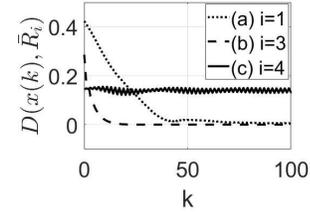}
			\caption{Simulation results of Example 2. (a) Asymptotically stable PIR ($\bar{R}_1$). (b) Asymptotically stable positively invariant radial 
				line ($\bar{R}_3$). (c) Stable PIS which is a union 
				of two rays ($\bar{R}_4$).}
			\label{fig:5-dim_system}
		\end{figure}
	}
	
	In this case, we identify the existence of four PISs $\bar{R}_1, \bar{R}_2, \bar{R}_3$ and $\bar{R}_4$. $\bar{R}_1$ and $\bar{R}_2$ are PIRs passing through the points  $\begin{bmatrix}
	-0.8524 & -0.4460 & -0.2333 & -0.1220 & -0.0724
	\end{bmatrix}^T$ and $\begin{bmatrix}
	-1 & -0.0003 & 0 & 0 & 0
	\end{bmatrix}^T$, respectively. $\bar{R}_1$ is a subset of a region 
	with corresponding IET $0.2060$. Whereas $\bar{R}_2$ is 
	a subset of a region with corresponding IET $0.0660$. 
	$\bar{R}_3$ is a positively invariant radial line passing through the 
	point $\begin{bmatrix}
	0 & 0 & -0.0045 & -0.0190 & 0.9998
	\end{bmatrix}^T$ in a region with corresponding IET 
	$0.2260$. $\bar{R}_3$ is the eigensubspace corresponding to a real 
	negative eigenvalue of the respective $G$ matrix. $\bar{R}_4$ is a 
	union of two rays passing through the points $\begin{bmatrix}
 0.6241 & 0.3266 & 0.1737 & 0.1137 & -0.6788
	\end{bmatrix}^T$ and $\begin{bmatrix}
0.5831 & 0.3054 & 0.1574 & 0.0610 & 0.7336
	\end{bmatrix}^T$, respectively, in a region with corresponding IET $0.207385$.  Note that, these two rays belong to the span of two eigenvectors of the respective $G$ matrix corresponding to two real distinct eigenvalues with magnitude equal to $\rho(G)$.

According to Theorem~\ref{thm:asy_stable_ray}, we can analyze the 
stability of the intersection of these PISs 
with the unit sphere by studying the spectrum of the corresponding $G$ 
matrix. In this way, we see that both $\bar{R}_1 \cap B_1(0)$ and 
$\bar{R}_3\cap B_1(0)$ are asymptotically stable. Whereas 
$\bar{R}_2\cap B_1(0)$ is unstable. Figure~\ref{fig:5-dim_system}(a) and 
Figure~\ref{fig:5-dim_system}(b), respectively, present the convergence of the 
system state to the asymptotically stable positively subregions 
$\bar{R}_1 \cap B_1(0)$ and $\bar{R}_3 \cap B_1(0)$ for an initial 
condition sufficiently close to the respective subregion. In 
Figure~\ref{fig:5-dim_system}, \vspace{-1ex}
\begin{equation*}
  D(x(k),\bar{R}) \ldef \text{dist} \left( \frac{ x(k) }{ \norm{x(k)} 
  }, \ \bar{R} \cap B_1(0) \right) \vspace{-1ex}. 
\end{equation*}

Note that $\bar{R}_4$ does not belong to the class of PISs considered in Theorem~\ref{thm:asy_stable_ray}. 
But we can analyze the stability of $\bar{R}_4\cap B_1(0)$ in a similar 
way. In this case, the $G$ matrix corresponding to $\bar{R}_4$ is 
diagonalizable with $\sigma(G)=\{\lambda_i\}$, $i \in \{1,2,..,5\}$. 
We find that, $\lambda_1=-\lambda_2=\rho(G)$ and $|\lambda_i|<\rho(G)$ 
for all $i \in \{3,4,5\}$. Let $v_i$ denote an eigenvector of $G$ 
corresponding to an eigenvalue $\lambda_i$ for all $i \in 
\{1,2,..,5\}$. Note that $\bar{R}_4=\{\alpha u_1 \cup \alpha u_2: 
\alpha \ge 0\}$ where $u_1=\alpha_1v_1+\alpha_2v_2$ and 
$u_2=\alpha_1v_1-\alpha_2v_2$ for some  $\alpha_1, \alpha_2 \in \real$. 
As in the proof of Theorem~\ref{thm:asy_stable_ray}, without loss of 
generality, let us suppose that $G$ and hence $J\ldef 
\frac{G}{\lambda_1}$ are in real Jordan form. Now, by using similar 
arguments as in the proof of Theorem~\ref{thm:asy_stable_ray}, 
especially inequality~\eqref{eq:stability_inequality} and the following 
discussion, we can show that $\bar{R}_4\cap B_1(0)$ is stable.
Figure~\ref{fig:5-dim_system}(c) presents the distance between $\bar{R}_4\cap 
B_1(0)$ and a $\gamma$ map for an initial condition arbitrarily close 
to $\bar{R}_4$.

\vspace{-3ex}
\section{Conclusion} \label{sec:conclusion} In this paper, we analyzed
the evolution of IETs along the trajectories of linear
systems under RBSTRs. Under this control
method, studying steady state behavior of the IETs is
equivalent to studying the existence of a conic subregion, which is a
positively invariant set under the map that gives the evolution of the
state from one event to the next. We provided necessary conditions and sufficient
conditions for the existence of a PIS. We
also provided necessary and sufficient conditions for a PIS to be stable and asymptotically stable.  We
  extended this analysis to convergence of IETs to a
  given periodic sequence. We verified the proposed results through
numerical simulations. Future work includes analysis of the
  average and periodicity of the IETs generated by a
  general class of triggering rules as well as the design of
  triggering policies under scheduling constraints.

\vspace{-2ex}
\section*{Acknowledgements}

We thank the anonymous reviewers for their suggestions, which
  helped improve the paper significantly.
\vspace{-2ex}
\bibliographystyle{IEEEtran}
\bibliography{references}
\end{document}